\documentclass[10pt,twocolumn,twoside,english]{IEEEtran}
\usepackage{amssymb}
\usepackage{amsmath}
\usepackage{graphicx,wrapfig,times,amsthm,bm,amsmath,amsfonts,amssymb}
\usepackage{subfig}
\usepackage{threeparttable}
\usepackage[wide]{sidecap}
\usepackage{array}
\usepackage{algorithm}
\usepackage{algpseudocode}
\usepackage{lipsum,babel}
\usepackage{color}
\usepackage{setspace}

\usepackage{tikz}
\usetikzlibrary{arrows.meta}
\tikzset{%
	>={Latex[width=2mm,length=2mm]},
	base/.style = {rectangle, rounded corners, draw=black,
		minimum width=0.3cm, minimum height=0.3cm,
		text centered, font=\texttrademark},
	OptimalPolicies/.style= {base, text=green},
	Estimation/.style= {base, text=red},
	Stabilization/.style= {base, text=blue},
	RLAlgo/.style= {base, text=violet},
}

\newtheorem{thm}{Theorem}
\newtheorem{lem}{Lemma}
\newtheorem{deff}{Definition}
\newtheorem{prop}{Proposition}
\newtheorem{assum}{Assumption}

\newtheorem{remark}{Remark}
\newtheorem{example}{Example}

\def \det {\mathrm{det}}

\def \exp {\mathrm{exp}}

\def \R {\mathbb{R}}
\def \C {\mathbb{C}}
\def \tailconst {b_1}
\def \tailcoeff {b_2}
\def \tailexp {\alpha}

\def \cost {\mathcal{J}}
\def \para {\theta}

\def \ssconstant {\rho}

\newcommand{\Mnorm}[2]{{\left\vert\kern-0.25ex\left\vert\kern-0.25ex\left\vert #1 
		\right\vert\kern-0.25ex\right\vert\kern-0.25ex\right\vert}_{#2}}
\newcommand{\Opnorm}[3]{{\left\vert\kern-0.25ex\left\vert\kern-0.25ex\left\vert #1 
		\right\vert\kern-0.25ex\right\vert\kern-0.25ex\right\vert}_{#2 \to #3}}
\newcommand{\norm}[2]{{\left\vert\kern-0.25ex\left\vert #1 
		\right\vert\kern-0.25ex\right\vert}_{#2}}
\newcommand{\eigmax}[1]{\left| \lambda_{\max} \left( #1 \right)\right|}
\newcommand{\eigmin}[1]{\left| \lambda_{\min} \left( #1 \right)\right|}
\newcommand{\tr}[1]{\mathrm{tr} \left( #1 \right)}
\newcommand{\PP}[1]{\mathbb{P} \left(#1\right)}
\newcommand{\E}[1]{\mathbb{E} \left[#1\right]}

\newcommand{\samplesize}[3]{{N}_{}\left(#2,#3\right)}
\newcommand{\event}[1]{{\mathcal{#1}}}

\newcommand{\rank}[1]{\mathrm{rank}\left(#1\right)}

\newcommand{\innerproductmin}[2]{\psi_{#1} \left(#2\right)}
\newcommand{\innerproductminconstant}[1]{\psi_0}

\newcommand{\loss}[3]{\mathcal{L}_#1^{#2} \left(#3\right)}

\newcommand{\dimension}[2]{\mathrm{dim}_{#1} \left(#2\right)}

\newcommand{\paraspace}[1]{\Omega^{(#1)}}

\newcommand{\Kmatrix}[1]{K\left(#1\right)}
\newcommand{\Lmatrix}[1]{L\left(#1\right)}
\newcommand{\extendedLmatrix}[1]{\tilde{L}\left(#1\right)}

\newcommand{\optcost}[1]{\mathcal{J}^\star \left(#1\right)}
\newcommand{\instantcost}[1]{c_{#1}}

\newcommand{\trans}[1]{D_{#1}}
\newcommand{\esttrans}[1]{\hat{D}_{#1}}

 \onecolumn 

\newcommand{\edit}[1]{\textcolor{black}{#1}}

\begin{document}
\title{Finite Time Adaptive Stabilization of LQ Systems}
\author{Mohamad Kazem~Shirani Faradonbeh,
	Ambuj~Tewari,
	and~George~Michailidis}
\maketitle

\begin{abstract}
	Stabilization of linear systems with unknown dynamics is a canonical problem in adaptive control. Since the lack of knowledge of system parameters can cause it to become destabilized, an adaptive stabilization procedure is needed prior to regulation. Therefore, the adaptive stabilization needs to be completed in finite time. In order to achieve this goal, asymptotic approaches are not very helpful. There are only a few existing non-asymptotic results and a full treatment of the problem is not currently available.

	In this work, leveraging the novel method of {\em random} linear feedbacks, we establish high probability guarantees for finite time stabilization. Our results hold for remarkably general settings because we carefully choose a minimal set of assumptions. These include stabilizability of the underlying system and restricting the degree of heaviness of the noise distribution. To derive our results, we also introduce a number of new concepts and technical tools to address {\em regularity} and instability of the closed-loop matrix.
\end{abstract}

\section{Introduction} \label{Intro} \label{background}
\IEEEPARstart{W}{e} consider finite time stabilization of the following linear system. Given the initial state $x(0) \in \R^{p}$, for $t=0,1,\cdots$ we have
\begin{eqnarray}
x(t+1) &=& A_0x(t)+B_0u(t)+ w(t+1), \label{systemeq1}
\end{eqnarray}
where at time $t$, the vector $x(t) \in \R^p$ corresponds to the state (and output) of the system, $u(t) \in \R^r$ is the control action, and $\left\{ w(t) \right\}_{t=1}^\infty$ is a sequence of noise (i.e. random disturbance) vectors. The dynamics of the system, i.e. both the transition matrix $A_0 \in \mathbb{R}^{p \times p}$, as well as the input matrix $B_0 \in \mathbb{R}^{p \times r}$, are fixed but {\em unknown}. 

In order to stabilize the system, we need to design an adaptive procedure and establish finite time theoretical guarantees of stabilization. Once the system has been stabilized, we can use an adaptive regulation policy to minimize the cost function determined by the application. Hence, the stabilization procedure needs to be completed in a relatively short time period. There is an extensive literature providing {\em infinite} time analyses to adaptively stabilize a Linear-Quadratic (LQ) system \cite{lai1986asymptotically,guo1988convergence,chen1989convergence,lai1991parallel,guo1991aastrom}, whereas finite time
results are scarce and rather incomplete. This work aims to contribute to this limited literature on the subject.

The evolution of LQ systems is governed by linear dynamics, while the operating cost is a quadratic function of the state and the control signal. To deal with the uncertainty about the true matrices guiding the system's dynamics, a standard scheme is Certainty Equivalence (CE) \cite{bar1974dual}. Its prescription is to assume that the {\em estimated} parameters coincide with the true dynamics matrices. However, it is shown that inconsistency occurs with positive probability \cite{lai1982least,becker1985adaptive,kumar1990convergence}, which can lead to instability. This motivated modifying CE to Optimism in the Face of Uncertainty (OFU) \cite{lai1985asymptotically,campi1997achieving,campi1998adaptive,bittanti2006adaptive}. OFU prescribes to act as if an {optimistic} approximation of the true parameter is the one guiding the evolution of the system. 

Recent {\em finite time} analyses consider a restricted setting \cite{abbasi2011regret,ibrahimi2012efficient}, where the proposed adaptive stabilization procedure heavily relies on the following strong conditions. First, {\em controllability} and {\em observability} of the true dynamics matrices of the system are assumed. \edit{Second, the closed-loop transition matrix is required to have {\em operator norm} less than one. Note that the former does not imply the latter \cite{bertsekas1995dynamic}. Third, uncertainty about the true dynamics matrices is restricted to an a priori {\em known} bounded box in the space of $p \times (p+r)$ matrices.} Finally, the noise vectors are supposed to have a sub-Gaussian distribution with uncorrelated coordinates. 

\edit{We introduce our stabilization algorithm and establish finite time guarantees for it in Section \ref{Stabilizing the System}. Leveraging the novel method of {\em random} linear feedbacks, we address the four aforementioned limitations. Indeed, the first assumption (which imposes a computationally intractable constraint \cite{faradonbeh2016optimality}), as well as the third one (which requires possibly unavailable information) are not needed. Further, we relax the operator norm condition to the minimal assumption of stabilizability. Finally, the noise process is generalized to the remarkably larger class of heavy-tailed sub-Weibull distributions with possibly correlated coordinates. Note that unlike the operator norm, stability of matrices is not preserved by multiplication (i.e. the product of stable matrices can be unstable). This means that existing theoretical techniques \cite{bittanti2006adaptive,abbasi2011regret} used in addressing the stabilization problem fail to work when the operator norm is not less than one.}

\edit{To derive finite time guarantees of stabilization, new concepts and technical tools are needed to address the following issues:
\begin{enumerate}
	\item 
	Because of the unbounded growth of the state vectors \cite{lai1985asymptotic}, the classical results of persistent excitation \cite{green1985persistence} are not applicable.
 \item 
	Since the system is not fully stabilized yet, the closed-loop matrix can have eigenvalues both inside and outside the unit circle. Thus, the smallest (largest) eigenvalue of the Gram matrix scales linearly (exponentially) with time \cite{lai1983asymptotic,ourpaper}. This leads to the failure of the existing approaches which do not need the persistent excitation condition \cite{abbasi2011regret,lai1986extended}.
	\item 
	For unstable systems, it is shown that the normalized empirical covariance of the state vector is a random matrix \cite{nielsen2005strong,nielsen2006order}. So, in order to obtain reliable identification results, anti-concentration properties of random matrices need to be carefully examined \cite{lai1983note}. 
	\item 
	For accurate identification, one needs to ensure that the important condition of closed-loop {\em regularity} holds (see Definition \ref{regularity}). It is a necessary condition on the eigenvalues of magnitude larger than one \cite{nielsen2009singular}.
\end{enumerate}
} 

The remainder of the paper is organized as follows. The problem is rigorously formulated in Section \ref{Optimal Policies}. Then, in Section \ref{Estimation} we study the key identification results for unstable closed-loop dynamics as the cornerstone of the stabilization algorithm presented later on. Subsequently in Section \ref{Stabilizing the System}, results regarding the properties of random linear feedback are established. Finally, we propose the adaptive stabilization {\bf Algorithm \ref{stabilizationalgo}}, and show that it is guaranteed to return a high probability stabilizing set.

\subsection{Notation}
The following notation is used throughout this paper. For matrix $A \in \C^{p \times q}$, $A'$ is its transpose. When $p=q$, the smallest (respectively largest) eigenvalue of $A$ (in magnitude) is denoted by $\lambda_{\min} (A)$ (respectively $\lambda_{\max}(A)$) and the trace of $A$ is denoted by $\tr{A}$. For $\gamma \in \mathbb{R}, \gamma \geq 1, v \in \mathbb{C}^q$,
the $\gamma$-norm of vector $v$ is $\norm{v}{\gamma} = \left(\sum\limits_{i=1}^{q} \left|v_i\right|^\gamma \right)^{1/\gamma}$. 
Further, when $\gamma=\infty$, the norm is defined according to $\norm{v}{\infty} = \max \limits_{1 \leq i \leq q} |v_i|$. 

We also use the following notation for the operator norm of matrices. For $\beta, \gamma \in \left[1,\infty\right]$, and $A \in \C^{p \times q}$, define 
$$\Opnorm{A}{\gamma}{\beta} = \sup \limits_{v \in \C^{q} \setminus \{0\}} \frac{\norm{Av}{\beta}}{\norm{v}{\gamma}}.$$
Whenever $\gamma = \beta$, we simply write $\Mnorm{A}{\beta}$. To denote the dimension of manifold $\mathcal{M}$ over the field $F$, we use $\dimension{F}{\mathcal{M}}$. Finally, the sigma-field generated by random vectors $X_1,\cdots,X_n$ is denoted by $\sigma \left( X_1,\cdots,X_n \right)$. \edit{The notations $\para,\Kmatrix{\para}, \Lmatrix{\para}$, and $\extendedLmatrix{\para}$ are defined in Remark \ref{paradeff}, equations \eqref{ricatti2}, \eqref{ricatti1}, and Remark \ref{Ltildedeff}, respectively.}

\section{Problem Formulation} \label{Optimal Policies}
We start by discussing the adaptive stabilization problem that constitutes the primary focus of this work. As mentioned above, the corresponding adaptive policy for regulating the system (i.e. cost minimization) can be employed, once the stabilization is guaranteed. Results of this work can be used in the finite time analysis of adaptive regulation for LQ systems. Further, stabilization of linear systems is intimately related to a Riccati equation for the corresponding LQ system. Therefore, we comprehensively discuss the necessary preliminaries here. 

The stochastic evolution of the system is governed by the linear dynamics \eqref{systemeq1}, where $\left\{ w(t) \right\}_{t=1}^\infty$ are independent mean-zero noise vectors with full rank covariance matrix $C$: 
\begin{equation*}
\E {w(t)}=0, \:\: \E {w(t)w(t)'}=C, \:\:\: \eigmin{C}>0.
\end{equation*}
Generalizations of the established results to dependent noise vectors (i.e. martingale difference sequences) is rather straightforward. The true dynamics matrices $A_0,B_0$ are assumed to be stabilizable, as defined below.
\begin{deff}[Stabilizability \cite{bertsekas1995dynamic}] \label{stabilizability}
	$\left[A,B\right]$ is stabilizable if there exists $L \in \R ^{r \times p}$ such that $\eigmax{A+BL} < 1$. The linear feedback matrix $L$ is called a stabilizer for $\left[A,B\right]$.
\end{deff}
\edit{\begin{remark} \label{paradeff}
	For notational convenience, henceforth for $A \in \R^{p \times p}$, $B \in \R^{p \times r}$, we use $\para$ to denote $\left[A,B\right]$. Clearly, $\para \in \R^{p \times q}$, where $q=p+r$.
\end{remark}} 
We assume {\em perfect observations}, i.e. the operator can fully observe the sequence of state vectors. Next, suppose that $\instantcost{t}$ is the quadratic instantaneous cost function at time $t$: 
\begin{eqnarray}
\instantcost{t}&=& x(t)'Qx(t) + u(t)'Ru(t), \label{systemeq2}
\end{eqnarray}
which is defined according to the known positive definite cost matrices $Q \in \R^{p \times p}, R \in \R^{r \times r}$. An adaptive policy is a mapping which designs the control action according to the cost matrices, and the history of the system. That is, for all $t=0,1,\cdots$, the operator needs to determine $u(t)$ according to $Q, R$, $\left\{ x(i) \right\}_{i=0}^t, \left\{ u(j) \right\}_{j=0}^{t-1}$.

The following proposition shows that in order to stabilize a linear system, one can solve a \emph{Riccati} equation. A solution, is a positive semidefinite matrix $\Kmatrix{\para}$ satisfying \eqref{ricatti2}. 
\begin{table*}
	\begin{eqnarray} 
	\Kmatrix{\para} &=& Q + A'\Kmatrix{\para}A - A' \Kmatrix{\para}B \left(B'\Kmatrix{\para}B+R\right)^{-1} B'\Kmatrix{\para}A , \label{ricatti2} \\
	\Lmatrix{\para} &=& -\left(B'\Kmatrix{\para}B+R\right)^{-1} B'\Kmatrix{\para}A. \label{ricatti1}
	\end{eqnarray}
\end{table*}
For this purpose, we introduce a notation that simplifies certain expressions throughout this work. 
\begin{remark} \label{Ltildedeff}
	For arbitrary stabilizable $\para_1,\para_2 \in \R^{p \times q}$, let $\extendedLmatrix{\para_1}=\begin{bmatrix} I_p \\ \Lmatrix{\para_1} \end{bmatrix} \in \R^{q \times p}$. So, $\para_2 \extendedLmatrix{\para_1}=A_2+B_2\Lmatrix{\para_1}$.
\end{remark}
\begin{prop} \label{stabilizable}
	If $\para$ is stabilizable, \eqref{ricatti2} has a unique solution. Conversely, if \eqref{ricatti2} has a solution, $\Lmatrix{\para}$ defined by \eqref{ricatti1} is a stabilizer for the dynamics parameter $\para$; i.e. $\eigmax{\para \extendedLmatrix{\para}}<1$. 
\end{prop}
The proof of Proposition \ref{stabilizable} is provided in Appendix A, where the following cost minimization property of Riccati equations \eqref{ricatti2}, \eqref{ricatti1} is established as well. Assuming the system evolves according to \eqref{systemeq1}, the linear feedback $u(t) = \Lmatrix{\para_0} x(t)$ minimizes the expected average cost of the system of dynamics parameter $\para_0$. Namely, letting $\instantcost{t}$ be as \eqref{systemeq2}, in general it holds that
\begin{equation*}
\limsup \limits_{T \to \infty} \frac{1}{T} \sum \limits_{t=1}^T \E{\instantcost{t}} \geq \tr{\Kmatrix{\para_0}C},
\end{equation*}
where the linear feedback $u(t) = \Lmatrix{\para_0} x(t)$ attains the equality. An adaptive stabilization procedure is ignorant about the true parameter $\para_0$, and needs to estimate it. The following lemma addresses the stability if the actual system evolution parameter is $\para_0$, while the linear feedback $\Lmatrix{\para}$ is designed according to the approximation $\para$. The proof of Lemma \ref{SNeighborhood} can be found in Appendix B.
\begin{lem} [Stabilizing neighborhood] \label{SNeighborhood}
	There is $\epsilon_0>0$, such that for every stabilizable $\para$, if $\Mnorm{\para-\para_0}{2} < \epsilon_0$, then, $\para_0\extendedLmatrix{\para}$ is stable. 
\end{lem}

\section{Closed-loop Identification} \label{Estimation}
When applying linear feedback $L \in \R^{r \times p}$, the dynamics take the form $x(t+1)=\trans{}x(t)+w(t+1)$, where $\trans{}=A_0+B_0L$ is the {\em unstable} closed-loop transition matrix. Subsequently, we present results for the accurate identification of $\trans{}$ through the least-squares estimator. Observing the state vectors $\left\{x(t)\right\}_{t=0}^n$, for an arbitrary matrix $E \in \R^{p \times p}$ define the sum-of-squares loss function
$\loss{n}{}{E}=\sum\limits_{t=0}^{n-1} \norm{x(t+1)- E x(t)}{2}^2$. Then, the true closed-loop transition matrix $\trans{}$ is estimated by $\esttrans{n} $, which is a minimizer of the loss function; $\loss{n}{}{\esttrans{n}} = \min\limits_{E \in \R^{p \times p}} \loss{n}{}{E}$. To analyze the finite time behavior of the aforementioned identification procedure, the following is assumed for the tail-behavior of every coordinate of the noise vector.
\begin{assum}[Sub-Weibull distribution \cite{ourpaper}] \label{tailcondition} \label{2tailcondition}
	There are positive reals $\tailconst, \tailcoeff$, and $\tailexp$, such that for all $t\geq 1; 1 \leq i \leq p; y >0$,  
	\begin{equation*}
	\PP{\left|w_i(t)\right| > y} \leq \tailconst \: \exp \left(-\frac{y^\tailexp}{\tailcoeff}\right).
	\end{equation*}	
\end{assum}
Intuitively, smaller values of the exponent $\tailexp$ correspond to heavier tails for the noise distribution, and vice versa. Note that whenever $\tailexp <1$, the noise coordinates $w_i(t)$ do not need to have a moment generating function. Further, the noise coordinates can be either discrete or continuous random variables, and are not assumed to have a probability density function (pdf). Henceforth, the special case of bounded noise can be obtained from the presented results letting $\tailexp \to \infty$. 

Next, we define an important property of unstable transition matrices which is required in order to obtain accurate estimation results.
\begin{deff}[Regularity \cite{nielsen2009singular}] \label{regularity} \label{2regularity}
	$\trans{} \in \R^{p \times p}$ is regular if for any eigenvalue $\lambda$ of $\trans{}$ such that $\left|\lambda\right|>1$, the geometric multiplicity of $\lambda$ is one.
\end{deff}
Regularity implies that the eigenspace corresponding to $\lambda$ is one dimensional, and vice versa. There are other equivalent formulations for regularity. Indeed, $\trans{}$ is regular if and only if for any eigenvalue $\lambda$ such that $\left| \lambda \right|>1$, in the Jordan decomposition of $\trans{}$ there is only one block corresponding to $\lambda$, regardless of its algebraic multiplicity. Another equivalent formulation is that $\trans{}$ is regular, if and only if $\rank{\trans{}-\lambda I_p} \geq p-1$, for all $\lambda \in \C$, $\left|\lambda \right|>1$. For example, let $P_1,P_2 \in \C^{2 \times 2}$ be arbitrary invertible matrices, and assume 
\begin{equation*}
\trans{1}=P_1^{-1}\begin{bmatrix}
\rho & 1 \\ 0 & \rho
\end{bmatrix} P_1, \trans{2} = P_2^{-1} \begin{bmatrix}
\rho & 0 \\ 0 & \rho
\end{bmatrix} P_2,
\end{equation*}
are real $2 \times 2$ matrices, where $\rho \in \C$ satisfies $\left|\rho\right|>1$. Then, $\trans{1}$ {\em is} regular, but $\trans{2}$ {\em is not}.

\edit{In order to examine the accuracy of the least-squares estimation, we leverage existing finite time identification results for unstable dynamics \cite{ourpaper}. First, if the empirical covariance matrix $V_n=\sum\limits_{t=0}^{n-1} x(t)x(t)'$ is non-singular, one can write $\esttrans{n}= \sum\limits_{t=0}^{n-1} x(t+1)x(t)'V_n^{-1}$. Hence, the behavior of $V_n$ governs the estimation accuracy. For unstable $D$, an appropriately normalized $V_n$ is shown to be a random matrix \cite{ourpaper}. Thus, letting $\tilde{V}_n$ denote the normalized matrix, the accuracy of $\esttrans{n}$ depends on the stochastic lower bounds of $\tilde{V}_n$. Let $\innerproductmin{}{\delta}$ be the high probability lower bound of $\tilde{V}_n$; i.e. it is sufficiently small to satisfy $\PP{\eigmin{\tilde{V}_n} < \innerproductmin{}{\delta}} < \delta$. The following statement studies $\innerproductmin{}{\delta}$ based on anti-concentration results for sequences of random matrices \cite{lai1983note}.
\begin{prop} \cite{ourpaper} \label{PDcore}
	Suppose that $D$ is regular. In general, $\delta>0$ implies that $\innerproductmin{}{\delta}> 0$. Further, if $w(t_0)$ has a bounded pdf for some $t_0 \geq 1$, then for all $\delta>0$ we have $\innerproductmin{}{\delta} \geq \innerproductminconstant{\trans{0}} \delta$, where $\innerproductminconstant{\trans{0}} >0$ is a fixed constant.
\end{prop}
Theorem \ref{consistency} determines the time length the user should interact with the system, in order to collect sufficiently many observations for accurate identification of the unstable matrix $\trans{}$. The sample size is based on the constant $\ssconstant_{}$, for which the exact dependence on the noise parameters $\tailconst, \tailcoeff, \tailexp$, $\eigmin{C}, \eigmax{C}$, and the closed-loop matrix $\trans{}$ is available \cite{ourpaper}. Moreover, let $\overline{\lambda}_1 , \cdots, \overline{\lambda}_{k}$ (respectively $\underline{\lambda}_1, \cdots,\underline{\lambda}_{\ell}$) be the distinct eigenvalues of $\trans{}$ outside (respectively inside) the unit circle. Then, $\innerproductmin{}{\delta}$ depends on $\eigmin{C}$, $\min\limits_{1 \leq i \leq \ell} 1-\left| \underline{\lambda}_i \right|$, $\min\limits_{1 \leq i \leq k} \log \left| \overline{\lambda}_i \right|$, and $\min\limits_{1 \leq i < j \leq k} \log \left| \overline\lambda_i - \bar{\lambda}_j \right|$ \cite{ourpaper}. The constant $\innerproductminconstant{}$ depends on the upper bound of the pdf of $w(t_0)$ as well. The explicit specification of these dependencies is fully presented in \cite{ourpaper} and hence ommitted.} Next, let $\samplesize{\ref{consistency}}{\epsilon}{\delta}$ be large enough, such that $n \geq \samplesize{\ref{consistency}}{\epsilon}{\delta}$ implies
\begin{equation} \label{generalcasecondition}
\frac{n}{\left(\log n\right)^{4/\tailexp}} \geq \frac{\ssconstant_{}}{\epsilon^2} \left( \left( -\log \delta \right)^{1+4/\tailexp} - \log \innerproductmin{}{\delta} \right ).
\end{equation}
\begin{thm} [Unstable identification \cite{ourpaper}] \label{consistency}
	Suppose that $D$ is regular, and has no eigenvalue of unit size. As long as $n \geq \samplesize{\ref{consistency}}{\epsilon}{\delta}$, we have \begin{equation*}
	\PP{\Mnorm{\esttrans{n}-D}{2} \leq \epsilon} \geq 1-\delta.
	\end{equation*}
\end{thm}
Hence, by \eqref{generalcasecondition}, the probability $\delta$ of having an identification error of magnitude $\epsilon$, decays exponentially fast when $n$ grows. In the next section, we show that one can satisfy the assumptions of Theorem \ref{consistency} by applying random linear feedbacks to a stabilizable system with unknown dynamics parameters. 
\section{Stabilization Algorithm} \label{Stabilizing the System}
Although the true parameter $\para_0$ is unknown, according to Lemma \ref{SNeighborhood}, a stabilizing linear feedback $\Lmatrix{\para}$ can be designed, if one can find a stabilizing neighborhood $\paraspace{0}$, such that
\begin{equation} \label{algo1eq0}
\paraspace{0} \subset \left\{ \para \in \R^{p \times q}: \Mnorm{\para - \para_0}{2} \leq \epsilon_0 \right\}.
\end{equation}
Using Theorem \ref{consistency}, we establish that $\paraspace{0}$ can be estimated if one applies a random linear feedback to the system. Since in Theorem \ref{consistency} the closed-loop transition matrix needs to be regular with no eigenvalue of unit size, first we need to show that these conditions can be satisfied. Lemma \ref{regularityproof}, and Lemma \ref{unitrootexclusion} accomplish this, with no knowledge beyond stabilizability of $\left[A_0,B_0\right]$. Based on the properties of the distribution of a random linear feedback matrix $L$, the above lemmas provide general statements, which hold almost surely. Then, we present a finite time stabilizing algorithm, and prove that it will provide us the desired stabilizing neighborhood. To proceed, we define the following classes of probability distributions over real valued vectors and matrices.
\begin{deff}[Full rank distributions] \label{fullrankdistribution}
	Let $X$ be a random vector in $\R^m$. $X$ has a linearly full rank distribution if for any arbitrary hyperplane $\mathcal{P} \subset \R^m$, it holds that $\PP{X \in \mathcal{P}}=0$. Further, $X$ has a general full rank distribution, if for every manifold $\mathcal{M} \subset \R^m$ such that $\dimension{\R}{\mathcal{M}} \leq m-1$, it holds that $\PP{X \in \mathcal{M}}=0$.
\end{deff}
The following example illustrates the difference between the two types of full rank distributions defined above.
\begin{example}
	Let $Z \in \R^p$ be normally distributed, $Z \sim \mathcal{N}\left(\mu,\Sigma\right)$, with arbitrary mean $\mu \in \R^p$, and positive definite covariance matrix $\Sigma \in \R^{p \times p}$. Then, $Z$ has a general full rank distribution. Letting $Y=\left(Z/\norm{Z}{2}\right) \boldsymbol{1}_{\left\{ Z \neq 0 \right\}}$, 
	the random vector $Y$ has a linearly full rank distribution, but since it lives on the unit sphere, $Y$ does not have a general full rank distribution.
\end{example}
Random linear feedbacks with full rank distributions induce the desired properties to the closed-loop transition matrix, as we rigorously establish below. 
\begin{lem}[Closed-loop Regularity] \label{regularityproof}
	Assume $\left[A_0,B_0\right]$ is stabilizable. Let the columns of $L \in \R^{r \times p}$ be independent (but not necessarily identically distributed), with linearly full rank distributions. The matrix $A_0+B_0L$ is regular, with probability one.
\end{lem}
\begin{proof}[\bf \edit{Proof of Lemma \ref{regularityproof}}]
	Let the event $\event{G}$ be that $D=A_0+B_0L$ is irregular. We prove that for all $\lambda \in \C$, $\left|\lambda\right|\geq 1$, with probability one, $\rank{D-\lambda I_p} \geq p-1$.
	Note that according to the discussion after Definition \ref{regularity}, this implies $\PP{\event{G}}=0$. 
	
	First, let $Y_i \in \R^m,i=1,\cdots,m$ have linearly full rank distributions. Define $Y=\left[Y_1,\cdots, Y_m\right]$, and let $M\left(\lambda\right)$ be a $m \times m$ matrix, with all coordinates being real polynomials of $\lambda$. Let $f\left(\lambda\right)$ be a real polynomial of $\lambda$ as well. We show that 
	\begin{equation} \label{regularityeq1}
	\PP{\exists \lambda \in \C, f\left(\lambda\right) \neq 0 : \rank{Y-\frac{M\left(\lambda\right)}{f\left(\lambda\right)}}<m-1}=0.
	\end{equation}
	If $\rank{Y-M\left(\lambda_0\right)/f\left(\lambda_0\right)}<m-1$, letting $M\left(\lambda_0\right)/f\left(\lambda_0\right)=\left[M_1,\cdots,M_m\right]$, two of the vectors $Y_i-M_i, i=1,\cdots,m$, such as $Y_{m-1}-M_{m-1}, Y_m-M_m$, can be written as linear combinations of the others. There are finitely many values of $\lambda_0$ for which $Y_{m-1}-M_{m-1}$ is a linear combination of $ Y_1-M_1, \cdots, Y_{m-2}-M_{m-2}$, since for every such a $\lambda_0$, $\det \left(\tilde{Y}\right)=0$, where $\tilde{Y}$ is the square matrix whose columns are $ Y_1-M_1, \cdots, Y_{m-1}-M_{m-1}$, removing an arbitrary row. Note that $\det \left(\tilde{Y}\right)$ is a polynomial of $\lambda_0$, divided by $f\left(\lambda_0\right)$, and $f\left(\lambda_0\right) \neq 0$.
	
	Note that $\lambda_0$ is a deterministic function of $Y_1, \cdots, Y_m$. For every such $\lambda_0$, the dimension of the subspace $\mathcal{P}$ spanned by $ Y_1-M_1, \cdots, Y_{m-2}-M_{m-2}, M_m$ is at most $m-1$. Because $Y_m$ is independent of $ Y_1, \cdots, Y_{m-1}$, and $Y_m$ has a linearly full rank distribution, $\PP{Y_m \in \mathcal{P}}=0$; i.e. \eqref{regularityeq1} holds.
	
	Now, let $m=\rank{B_0}$. If $m=p$, applying the above argument to $Y=D, M\left(\lambda\right)=\lambda I_p, f\left(\lambda\right)=1$,
	we have $\PP{\event{G}}=0$, since full rankness of $B_0$ implies linearly full rank distributions for all columns of $B_0L$. If $m<p$, there is a $p \times p$ permutation matrix $J$, and $K \in \R^{\left(p-m\right) \times m}$, such that
	$JB_0=\begin{bmatrix}
	\tilde{B} \\ K \tilde{B}
	\end{bmatrix}= \begin{bmatrix}
	I_m \\ K
	\end{bmatrix} \tilde{B}$,
	where $\tilde{B} \in \R^{m \times r}$ is full rank. Let $L_0$ be a stabilizer, $D_0=A_0+B_0L_0$, and $JD_0= \begin{bmatrix} D_1 \\ D_2\end{bmatrix}, D_1 \in \R^{m \times p}, D_2 \in \R^{\left(p-m\right)\times p}$,
	to get
	$J \left(A_0+B_0L\right)= \begin{bmatrix}
	D_1 + \tilde{B}\left(L-L_0\right) \\ D_2 + K \tilde{B}\left(L-L_0\right)
	\end{bmatrix}$.
	Writing 
	$J= \begin{bmatrix} J_1 \\ J_2\end{bmatrix}, J_1 \in \R^{m \times p}, J_2 \in \R^{\left(p-m\right)\times p}$, 
	we have
	\begin{eqnarray*}
		&& \rank{A_0+B_0L - \lambda I_p}\\
		&=& \rank{\begin{bmatrix}
				I_m & 0_{m \times \left(p-m\right)} \\
				-K & I_{p-m}
			\end{bmatrix} \left(J \left(A_0+B_0L\right)-\lambda J\right)} \\
		&=& \rank{\begin{bmatrix}
				D_1 + \tilde{B}\left(L-L_0\right) - \lambda J_1 \\
				\left[-K,I_{p-m}\right] J \left(D_0-\lambda I_p \right)
		\end{bmatrix}}.
	\end{eqnarray*}
	
	Denote the last matrix above by $\tilde{X}$. Since $\eigmax{D_0}<1$, for $\left|\lambda\right| \geq 1$ the matrix $D_0-\lambda I_p$ is full rank. Therefore, because of $\rank{\left[-K,I_{p-m}\right]}=p-m$,
	we have $\rank{\left[-K,I_{p-m}\right] J \left(D_0-\lambda I_p \right)} = p-m$. 
	
	Rearrange the columns of matrix $\tilde{X}$ to get 
	$X=\begin{bmatrix} X_{11} & X_{12} \\ X_{21} & X_{22}\end{bmatrix}$,
	such that $X_{11} \in \C^{m \times m}, X_{22} \in \C^{\left(p-m\right) \times \left(p-m\right)}, \rank{X_{22}}=p-m$. In other words, $p-m$ linearly independent columns of $\left[-K,I_{p-m}\right] J \left(D_0-\lambda I_p \right)$ have been put together to form $X_{22}$. If $D$ is not regular, 
	\begin{eqnarray*}
		p-2 &\geq& \rank{\tilde{X}} = \rank{X} \\
		&=& \rank{X \begin{bmatrix}
				I_m & 0_{m \times \left(p-m\right)} \\
				-X_{22}^{-1}X_{21} & I_{p-m}
		\end{bmatrix}} \\
		&=& \rank{\begin{bmatrix}
				X_{11}-X_{12}X_{22}^{-1}X_{21} & X_{12} \\
				0_{\left(p-m\right) \times m} & X_{22}
		\end{bmatrix}}.
	\end{eqnarray*}				
	
	Hence, $\rank{X_{11}-X_{12}X_{22}^{-1}X_{21}} \leq m-2$.
	Recall that columns of $\left[X_{11},X_{12}\right]$ are exactly the same as $D_1 + \tilde{B}\left(L-L_0\right) - \lambda J_1$, and all coordinates of $\det \left(X_{22}\right) X_{12}X_{22}^{-1}X_{21}$ are polynomials of $\lambda$ (since all coordinates of $\det \left(X_{22}\right) X_{22}^{-1}$ are polynomials of the coordinates of $X_{22}$). Taking $ f\left(\lambda\right)= \det \left(X_{22}\right)$,
	by \eqref{regularityeq1}, since full rankness of $\tilde{B}$ implies linearly full rank distributions for all columns of $\tilde{B}\left(L-L_0\right)$, we have $\PP{\rank{X_{11}-X_{12}X_{22}^{-1}X_{21}} \leq m-2} =0$,
	which is the desired result since $\rank{X_{22}}=p-m$. 
\end{proof}

If the distribution of linear feedback $L$ is generally full rank, the following results shows that $A_0+B_0L$ has no eigenvalue on the unit circle of the complex plane.
\begin{lem}[Closed-loop Eigenvalues] \label{unitrootexclusion}
	Assume $\left[A_0,B_0\right]$ is stabilizable. Let $L \in \R^{r \times p}$ have a general full rank distribution over $\R^{r \times p}$. With probability one, $A_0+B_0L$ has no unit size eigenvalue.
\end{lem}
\begin{proof}[\bf \edit{Proof of Lemma \ref{unitrootexclusion}}]
	Assume $D=A_0+B_0L$ has a unit-root eigenvalue, denoted by $\lambda \in \C, \left| \lambda \right|=1$. Further, assume that $m=\rank{B_0}$, and let the permutation matrix $J$ and the matrix $K \in \R^{\left(p-m\right) \times m}$ be such that
	\begin{equation*}
	JB_0=\begin{bmatrix}
	\tilde{B} \\ K \tilde{B}
	\end{bmatrix}= \begin{bmatrix}
	I_m \\ K
	\end{bmatrix} \tilde{B},
	\end{equation*} 
	where $\tilde{B} \in \R^{m \times r}$ is full rank. Letting $L_0$ be a stabilizer, $D_0=A_0+B_0L_0$, and $X=\tilde{B} \left(L-L_0\right) \in \R^{m \times p}$, note that $X$ has a general full rank distribution, thanks to the full-rankness of $\tilde{B}$. Since $D_0$ is stable, $\det \left(D_0 - \lambda I_p\right) \neq 0$, and
	\begin{eqnarray*}
		0 &=& \det \left(A_0+B_0L - \lambda I_p\right) \\
		&=& \det \left(JD_0 + \begin{bmatrix} I_m \\ K \end{bmatrix} X - \lambda J\right) \\
		&=& \det \left( \left(D_0 - \lambda I_p\right)^{-1} J^{-1} \begin{bmatrix} I_m \\ K \end{bmatrix} X + I_p \right) \\
		&=& \det \left( X \left(D_0 - \lambda I_p\right)^{-1} J^{-1} \begin{bmatrix} I_m \\ K \end{bmatrix} + I_m \right),
	\end{eqnarray*}
	where the last equality above is implied by Sylvester's determinant identity. Denote the complex conjugate of $\lambda$ by $\bar{\lambda}$, and define the real matrix
	\begin{equation*}
	M \left({\lambda}\right) = M \left(\bar{\lambda}\right)= \left(D_0 - \bar{\lambda} I_p\right)^{-1} \left(D_0 - \lambda I_p\right)^{-1} J^{-1} \begin{bmatrix} I_m \\ K \end{bmatrix}.
	\end{equation*}
	Further, define the space of eigenvectors in $\C^m$ as follows. First, consider the relation $\sim$ on $\C^m$, defined as $$x \sim y \text{, if } x=cy \text{ for some } c \in \C, c \neq 0.$$ 
	Since $\sim$ is an equivalence relation, for the set of equivalence classes denoted by $S= \frac{\C^m }{\sim}$ (which is the direction space in $\C^m$) we have $\dimension{\C}{S}=m-1$; i.e. $\dimension{\R}{S}=2m-2$.
	
	Note that for every matrix $Y \in \C^{m \times m}$ and every vector $v \in \C^m$, $Yv=0$ if and only if $Y \tilde{v}=0$ for every $\tilde{v} \sim v$. Thus, $\det \left( X \left(D_0 - \lambda I_p\right)^{-1} J^{-1} \begin{bmatrix} I_m \\ K \end{bmatrix} + I_m \right)=0$ implies that there is $v \in S, v \neq 0$, such that
	\begin{equation} \label{unitrootexclusioneq1}
	\left( X \left(D_0 - \bar{\lambda} I_p\right) M \left(\lambda\right) + I_m \right)v=0
	\end{equation}
	Denote the set of all matrices $X$ satisfying \eqref{unitrootexclusioneq1} by $\mathcal{X} \left(\lambda, v \right) \subset \R^{m \times p}$. Separating the real ($\Re$) and imaginary ($\Im$) parts, we get $X a(v) = \Re\left(v\right)$, $X b(v) = \Im \left(v\right)$, where for $v \in S$, the vectors $a(v),b(v) \in \R^p$ are defined as
	\begin{eqnarray*}
		a(v) &=&  M \left(\lambda\right) \Re\left(\bar{\lambda} v\right) - D_0 M\left(\lambda\right) \Re\left(v\right), \\
		b(v) &=& M \left(\lambda\right) \Im\left(\bar{\lambda} v\right) - D_0 M\left(\lambda\right) \Im\left(v\right).
	\end{eqnarray*}
	Next, we partition $S$ to $S_1, S_2$; i.e. $S=S_1 \cup S_2, S_1 \cap S_2 = \emptyset$, where
	\begin{eqnarray*}
		S_1 &=& \{ v \in S : a(v), b(v) \text{ are in-line } \}, \\
		S_2 &=& \{ v \in S : a(v), b(v) \text{ are not in-line } \}. 
	\end{eqnarray*}
	Whenever $v \in S_2$, for $j=1,\cdots,m$, the $j$-th row of $X$ needs to be in the intersection of two nonparallel hyperplanes $\mathcal{P}_1,\mathcal{P}_2 \subset \R^p$, where
	\begin{eqnarray*}
		\mathcal{P}_1 &=& \left\{ y \in \R^p: y'a(v)=\Re(v_j) \right\}, \\
		\mathcal{P}_2 &=& \left\{ y \in \R^p: y'b(v)=\Im(v_j) \right\}.
	\end{eqnarray*}
	Since $\dimension{\R}{\mathcal{P}_1} \leq p-1$, $\dimension{\R}{\mathcal{P}_2} \leq p-1$, and $v \in S_2$, we have $ \dimension{\R}{\mathcal{P}_1 \cap \mathcal{P}_2} \leq p-2$.
	Therefore, for $v \in S_2$, we have $\dimension{\R}{\mathcal{X} \left(\lambda, v \right)} \leq m(p-2)$. Since $\dimension{\R}{\left| \lambda \right|=1}=1$, using $\dimension{\R}{S_2} \leq 2m-2$ we have
	\begin{equation} \label{unitrooteq1}
	\dimension{\R}{\mathcal{Z}_1} \leq 1+ 2m-2+ m(p-2)= mp-1,
	\end{equation}
	where $\mathcal{Z}_1 = \bigcup\limits_{\left|\lambda\right|=1, v \in S_2} \mathcal{X} \left(\lambda,v\right)$.
	
	On the other hand, for $v \in S_1$, there is a real number, say $\varphi(v)$, such that $b(v)=\varphi(v) a(v)$. Then,
	\begin{equation} \label{unitrooteq2}
	\Im \left(v\right) = X b(v) = \varphi(v) X a(v) = \varphi(v) \Re \left(v\right),
	\end{equation} 
	i.e. whenever $v \in S_1$, the vectors $\Re(v), \Im(v)$ are in-line. So, $\dimension{\R}{S_1}=m-1$, and for $v \in S_1$, we have $\mathcal{P}_1 = \mathcal{P}_2$, i.e. $\dimension{\R}{\mathcal{X} \left(\lambda, v \right)} \leq m(p-1)$.
	After doing some algebra, we obtain
	\begin{eqnarray*}
		0 &=& \varphi(v) a(v) - b(v) \\
		&=& \varphi(v) \left(\Re\left(\lambda\right)I_p+ \varphi(v)\Im\left(\lambda\right)I_p - D_0\right)M\left(\lambda\right)\Re\left(v\right) \\
		&-& \left(\varphi(v)\Re\left(\lambda\right)I_p - \Im\left(\lambda\right)I_p - \varphi(v)D_0\right)M\left(\lambda\right)\Re\left(v\right) \\
		&=& \left(1+\varphi(v)^2\right) \Im \left(\lambda\right) M\left(\lambda\right)\Re\left(v\right),
	\end{eqnarray*}
	i.e. either $\Im \left(\lambda\right)=0$, or $M\left(\lambda\right)\Re\left(v\right)=0$. According to the definition of $M\left(\lambda\right)$, the latter case implies $\Re\left(v\right)=0$, which due to \eqref{unitrooteq2} leads to $v=0$, and is impossible. So,
	by $\dimension{\R}{\left| \lambda \right|=1, \Im(\lambda)=0}=0$, we have
	\begin{equation} \label{unitrooteq3}
	\dimension{\R}{\mathcal{Z}_2} \leq m-1+m(p-1) = mp-1,
	\end{equation}
	where $\mathcal{Z}_2 = \bigcup\limits_{\left|\lambda\right|=1, \Im\left(\lambda\right)=0} \mathcal{X} \left(\lambda,v\right)$. Writing $\mathcal{X}= \bigcup\limits_{\left|\lambda\right|=1, v \in S} \mathcal{X} \left(\lambda,v\right) \subset \mathcal{Z}_1 \cup \mathcal{Z}_2$,
	according to \eqref{unitrooteq1}, \eqref{unitrooteq3} we have $\dimension{\R}{\mathcal{X}} \leq mp-1$, and by general full-rankness of the distribution of $X$, the desired result holds: $\PP{\mathcal{X}}=0$.
\end{proof}

Subsequently, an algorithmic procedure to find a stabilizing neighborhood will be presented based on random linear feedbacks discussed above. First, letting $k=1+\lceil \frac{r}{p} \rceil$,  draw the columns of $L_1, \cdots, L_k \in \R^{r \times p}$ from independent standard Gaussian distributions $\mathcal{N}\left(0,I_r\right)$. Note that because of independence, for all $i=1,\cdots,k$, the random feedback $L_i$ has a general full rank distribution $\R^{r \times p}$. Lemma \ref{regularityproof} and Lemma \ref{unitrootexclusion} show that the conditions of Theorem \ref{consistency} hold. Therefore, every closed-loop transition matrix $D^{(i)}=A_0+B_0L_i$ can be estimated arbitrarily accurate. We show how to find a high probability confidence set for $\para_0$, using the accurate estimates of $D^{(1)},\cdots, D^{(k)}$. 
\begin{algorithm}
	\caption{{\bf: Adaptive Stabilization} \\ Output: Stabilizing Set $\paraspace{0}$ } \label{stabilizationalgo}
	\begin{algorithmic}
		\State Let $k=1+\lceil \frac{r}{p} \rceil, \tau_0=0$
		\For{$i=1,\cdots,k$}
		\For{$j=1,\cdots,p$}
		\State Draw column $j$ of $L_i$ from $\mathcal{N}\left(0,I_r\right)$, independently
		\EndFor
		\EndFor
		\State Define $M, \tilde{\epsilon}$ according to \eqref{algo1eq1}, \eqref{algo1eq2}, respectively
		\For{$i=1,\cdots,k$}
		\State Define $\tau_i$ by \eqref{algo1eq3}
		\While{$t < \tau_i$}
		\State Apply control action $u(t)=L_i x(t)$
		\EndWhile
		\State Estimate $\hat{D}^{(i)}$ by \eqref{algo1eq5}
		\State Construct $\paraspace{i}$ by \eqref{algo1eq6}
		\EndFor
		\State \textbf{return} $\paraspace{0} = \bigcap \limits_{i=1}^k \paraspace{i}$
	\end{algorithmic}
\end{algorithm}

Letting $\epsilon_0$ be as Lemma \ref{SNeighborhood}, define the precision $\tilde{\epsilon}$ and the matrix $M$ containing all matrices $L_1, \cdots, L_k$ by
\begin{eqnarray} 
M &=& \begin{bmatrix} 
I_p & \cdots & I_p \\ L_1 & \cdots & L_k 
\end{bmatrix} \in \R^{q \times kp}, \label{algo1eq1} \\
\tilde{\epsilon} &=& \frac{\epsilon_0}{2k} \inf \left\{ \frac{\Mnorm{\para M}{2}}{\Mnorm{\para}{2}} : \para \in \R^{p \times q} \right\}. \label{algo1eq2}
\end{eqnarray}
As a matter of fact, since $kp \geq q$ and $L_1,\cdots,L_k$ are independent, we have $\rank{M}=q$, almost surely; i.e. $\PP{\tilde{\epsilon}>0}=1$. Further, if $\tilde{\epsilon}$ became too small once $L_1,\cdots,L_k$ are drawn, one can repeatedly draw the random feedbacks to avoid pathologically small values of $\tilde{\epsilon}$. 

\edit{Then, let $\tau_0=0$, and for $i=1,\cdots,k$ define the following:   
\begin{eqnarray}
\tau_i &=&\tau_{i-1}+ \samplesize{\ref{consistency}}{\tilde{\epsilon}}{\frac{\delta}{k}}, \label{algo1eq3} \\
\hat{D}^{(i)} &=& \arg\min \limits_{E \in \R^{p \times p}} \sum\limits_{t=\tau_{i-1}}^{\tau_i-1} \norm{x(t+1)- E x(t)}{2}^2, \label{algo1eq5} \\
\paraspace{i} &=& \left\{ \para \in \R^{p \times q}: \Mnorm{\para \begin{bmatrix} I_p \\ L_i \end{bmatrix} - \hat{D}^{(i)}}{2} \leq \tilde{\epsilon} \right\}, \label{algo1eq6}
\end{eqnarray}
where the sample size $\samplesize{\ref{consistency}}{\cdot}{\cdot}$ is given in \eqref{generalcasecondition}. Conceptually, $\tau_i$ is the time point when the control action changes, $\hat{D}^{(i)}$ is the least-squares estimate, and $\paraspace{i}$ is a confidence set for $\para_0$. In fact, for each $i=1,\cdots, k$, Algorithm \ref{stabilizationalgo} applies the linear feedback $u(t)=L_i x(t)$ during the time period $\tau_{i-1} \leq t < \tau_i$. Then, observing $\left\{ x(t) \right\}_{t=\tau_{i-1}}^{\tau_i-1}$, the algorithm uses $\hat{D}^{(i)}$ to estimate the true closed-loop matrix $D^{(i)}=A_0+B_0L_i = \para_0 \begin{bmatrix} I_p \\ L_i \end{bmatrix}$. Finally, the high probability confidence set $\paraspace{i}$ is constructed for the true parameter $\para_0$, according to $L_i$. Iterating the above procedure for all $1 \leq i \leq k$, the algorithm constructs $\paraspace{1}, \cdots, \paraspace{k}$, and returns $\paraspace{0}=\bigcap\limits_{i=1}^k \paraspace{i}$ as an stabilizing set. Below, we show that it satisfies \eqref{algo1eq0}.}

\edit{By Theorem \ref{consistency}, \eqref{algo1eq3} implies that $\Mnorm{\hat{D}^{(i)}-D^{(i)}}{2} \leq \tilde{\epsilon}$, with probability at least $1-\delta/k$. So, by \eqref{algo1eq6}, we have $\PP{\para_0 \notin \paraspace{i}} \leq \delta/k$; i.e. $\PP{\para_0 \notin \paraspace{0}} \leq \delta$. To show that $\paraspace{0}$ is a stabilizing set, let $\para_1 \in \paraspace{0}$ be arbitrary. On the event $\para_0 \in \paraspace{0}$, for all $i=1,\cdots,k$ we have $\Mnorm{\left( \para_1 - \para_0 \right) \begin{bmatrix} I_p \\ L_i \end{bmatrix} }{2} 
\leq 2 \tilde{\epsilon}$. Using the definition of $M$ in \eqref{algo1eq1}, the latter result leads to $\Mnorm{\left( \para_1 - \para_0 \right) M }{2} \leq 2k\tilde{\epsilon} $. Thus, \eqref{algo1eq2} implies the following for $\para_1 \neq \para_0$:
\begin{equation*}
\frac{2k \tilde{\epsilon}}{\Mnorm{\para_1 - \para_0 }{2}} \geq
\frac{\Mnorm{\left( \para_1 - \para_0 \right) M }{2}}{\Mnorm{\para_1 - \para_0 }{2}} \geq \frac{2k \tilde{\epsilon}}{\epsilon_0},
\end{equation*}
or equivalently $\Mnorm{\para_1 - \para_0 }{2} \leq \epsilon_0$, which is the desired inequality of \eqref{algo1eq0}. Note that since $\para_0 \in \paraspace{0}$, with probability at least $1-\delta$, the failure probability of Algorithm \ref{stabilizationalgo} is at most $\delta$. This completes the proof of the following result.} 
\begin{thm}[Stabilization] \label{stabilizationtheorem}
	Let $\paraspace{0}$ be the stabilizing set provided by Algorithm \ref{stabilizationalgo}. For arbitrary $\para \in \paraspace{0}$, we have
	\begin{equation*}
	\PP{\eigmax{\para_0 \extendedLmatrix{\para}}<1} \geq 1- \delta.
	\end{equation*}
\end{thm}
In other words, the probability of failing to stabilize the system decays exponentially when the time of interaction with the system grows (see \eqref{generalcasecondition}). Obviously, the normal distribution $\mathcal{N}\left(0,I_r\right)$ used in Algorithm \ref{stabilizationalgo} is not unique, and can be substituted by any general full rank distribution over $\R^r$.
 
\section{Conclusion} \label{future}

We studied an adaptive stabilization scheme for linear dynamical systems, focusing on finite time analysis. Tailoring a novel procedure based on random linear feedbacks, we established non-asymptotic results under mild assumptions, namely those of
system stabilizability and a fairly general noise process that encompasses heavy-tailed distributions.

There are a number of interesting extensions of the current work. First, finite time analysis of stabilization given {\em noisy observations} of the state vector is an interesting topic for future investigation. Second, studying the stabilization problem in a high-dimensional setting (assuming sparsity or some other low dimensional structure) is also an interesting subject to be addressed in the future.


\appendices
\section{\edit{Proof of Proposition \ref{stabilizable}}}
\begin{proof}
	For convenience, let $K_0=\Kmatrix{\para_0},$ and $L_0=\Lmatrix{\para_0}$. First, assume $\left[A_0,B_0\right]$ is stabilizable, $L$ is a stabilizer, $D=A_0+B_0L$, and $\eigmax{D} < 1$. For arbitrary fixed PSD matrix $P_0$, define $P_t\left(P_0\right), t=1, \cdots, T$ recursively,
	\begin{eqnarray*}
	&& P_t\left(P_0\right) = Q+A_0'P_{t-1}\left(P_0\right)A_0\\ 
	&-& A_0'P_{t-1}\left(P_0\right)B_0 \left(B'_0 P_{t-1}\left(P_0\right) B_0+R\right)^{-1} B'_0P_{t-1}\left(P_0\right)A_0. \label{Pdefine}
	\end{eqnarray*}	
	Letting $\instantcost{t}$ be as defined in \eqref{systemeq2}, the optimal control policy for minimizing the finite horizon cumulative cost $\cost_T = \sum\limits_{t=0}^{T-1} \E{c_t} + \E{x(T)'P_0x(T)}$, 
	is $u(t)=L_t x(t), t=0,\cdots,T-1$, \cite{bertsekas1995dynamic}, where
	\begin{equation} \label{optimalfeedback}
	L_t=- \left(B_0' P_{T-t-1}\left(P_0\right) B_0+R\right)^{-1}B_0'P_{T-t-1}\left(P_0\right)A_0 .
	\end{equation}
	Moreover, this optimal policy yields the optimal cost
	\begin{eqnarray}
	\min \cost_T= x(0)'P_T\left(P_0\right) x(0) + \sum\limits_{t=0}^{T-1} \tr {CP_t\left(P_0\right)}. \label{mincost1}
	\end{eqnarray}
	On the other hand, applying the control policy $u(t)=Lx(t), 0 \leq t \leq T-1$, we have 
	\begin{eqnarray*}
	\E{x(T)'P_0x(T)|x(T-1)} &=& x(T-1)'D'P_0 Dx(T-1) \\
	&+& \tr {C P_0}, \\
	\E {c_{t+1}|x(t)} &=& x(t)'D'\left(Q+L'RL\right)Dx(t) \\
	&+& \tr {C \left(Q+L'RL\right)},
	\end{eqnarray*}
	for $t=0, \cdots, T-2$. Hence, the finite horizon cost becomes
	\begin{eqnarray}
	\cost_T= x(0)'\tilde{P}_T\left(P_0\right) x(0) + \sum\limits_{t=0}^{T-1} \tr {C\tilde{P}_t\left(P_0\right)}, \label{mincost2}
	\end{eqnarray}
	where $\tilde{P}_t\left(P_0\right), t=1, \cdots, T$ are defined recursively as
	\begin{eqnarray}
	\tilde{P}_0\left(P_0\right) &=& P_0, \\
	\tilde{P}_t\left(P_0\right) &=& Q+L'RL+ D'\tilde{P}_{t-1}\left(P_0\right)D. 
	\end{eqnarray}
	Since $ \eigmax{D} < 1$, $\lim\limits_{T \to \infty} \tilde{P}_T\left(P_0\right) = P_\infty$ for a PSD matrix $P_\infty$. Letting $C \to 0$, by \eqref{mincost1}, \eqref{mincost2} we have $$x(0)'P_T\left(P_0\right)x(0) \leq x(0)'\tilde{P}_T\left(P_0\right)x(0),$$ 
	i.e. $x(0)'P_T\left(P_0\right)x(0), T=1,2,\cdots$ is bounded. If $P_0=0$, this sequence is nondecreasing, because minimizing both sides of $$\sum\limits_{t=0}^{T-1} c_t  \leq \sum\limits_{t=0}^{T} c_t$$
	subject to $$x(t+1)=A_0x(t)+B_0u(t),$$ 
	we get
	\begin{equation*}
	x(0)'P_T\left(0\right) x(0) \leq x(0)'P_{T+1}\left(0\right) x(0).
	\end{equation*} 
	Therefore, the nondecreasing bounded sequence $\left\{x(0)'P_T\left(0\right)x(0)\right\}_{T=1}^\infty$ converges. Since $x(0)$ is arbitrary, $\left\{P_T(0)\right\}_{T=1}^\infty$ itself converges: $$\lim\limits_{T \to \infty} P_T(0)= P_\infty(0).$$ 
	According to the recursive definition of $P_t(0)$ in \eqref{Pdefine}, $P_\infty(0)$ is a solution of \eqref{ricatti2}. This shows the existence of a solution, while uniqueness will be established later. 
	
	Next, since $\lim\limits_{T \to \infty} P_T(0)= P_\infty(0)$, \eqref{mincost1} implies $\lim\limits_{t \to \infty} \tr{CP_t(0)}=\tr{CP_\infty(0)}$. So, the Cesaro mean also converges to this limit, i.e. $$\optcost{\para_0}=\tr{CP_\infty(0)}.$$ 
	Optimality of the linear feedback $u(t)=L_0x(t)$, is then established through \eqref{optimalfeedback}. Now, we are ready to show that $L_0$ is a stabilizer. Letting $$D=A_0+B_0L_0, C \to 0, K_0 = P_\infty(0),$$ 
	we show that for arbitrary $x(0)$, $x(t)=D^tx(0)$ vanishes as $t$ grows. First, note that by \eqref{ricatti2}, \eqref{ricatti1},
	\begin{eqnarray*}
		\left(B_0'K_0B_0+R\right) L_0 &=& -B_0'K_0A , \\
		L_0' \left(B_0'K_0B_0+R\right) L_0 &=& A_0' K_0B_0 \left(B_0'K_0B_0+R\right)^{-1} B_0'K_0A_0 .
	\end{eqnarray*}
	Therefore, we obtain
	\begin{eqnarray*}
		&& Q+L_0'RL_0+D'K_0D \\
		&=& Q+ A_0'K_0A_0+L_0'\left(B_0'K_0A_0+R \right)L_0 \\
		&+& A_0'K_0B_0L_0+ L_0'B_0'K_0A_0 \\
		&=& Q+ A_0'K_0A_0-A_0' K_0B_0 \left(B_0'K_0B_0+R\right)^{-1} B_0'K_0A_0 \\
		&+& \left[L_0'\left(B_0'K_0B_0+R\right)+A_0'K_0B_0\right]L_0 \\
		&+& L_0'\left[\left(B_0'K_0B_0+R\right) L_0+B_0'K_0A_0\right] 
		= K_0, 
	\end{eqnarray*}
	that is, 
	\begin{equation} \label{ricatti3}
	K_0-D'K_0 D = Q+L_0'RL_0.
	\end{equation}
	So, 
	\begin{equation} \label{stable1}
	x(t+1)'K_0 x(t+1)-x(t)'K_0 x(t) = -x(t)'\left(Q+L_0'RL_0\right) x(t).
	\end{equation}
	Adding up both sides of \eqref{stable1}, because $K_0$ is PSD we get
	\begin{eqnarray} 
	-x(0)'K_0 x(0) &\leq& x(t+1)'K_0 x(t+1)-x(0)'K_0 x(0) \notag \\
	&=& -\sum\limits_{i=0}^{t} x(i)'\left(Q+L_0'RL_0\right) x(i). \label{stable2}
	\end{eqnarray}
	In other words, $$\lim\limits_{t \to \infty} x(t)'\left(Q+L_0'RL_0\right) x(t) = 0.$$
	Thus, since $Q$ is positive definite, $\lim\limits_{t \to \infty} x(t) = 0$, i.e. $L_0$ is a stabilizer. Back to the proof of the existence of a solution $K_0$, we show that for arbitrary PSD $P_0$, it holds that $\lim\limits_{T \to \infty} P_T(P_0)= P_\infty(0)$. To do so, minimizing both sides of $$\sum\limits_{t=0}^{T-1} c_t  \leq \sum\limits_{t=0}^{T-1} c_t + x(T)'P_0x(T),$$
	subject to $$x(t+1)=A_0x(t)+B_0u(t),$$ 
	we get
	\begin{equation} \label{arbitrary1}
	x(0)'P_T\left(0\right) x(0) \leq x(0)'P_T\left(P_0\right) x(0).
	\end{equation} 
	On the other hand, applying controller $u(t)=L_0x(t)$, the cost $\sum\limits_{t=0}^{T-1} c_t + x(T)'P_0x(T)$ becomes
	\begin{equation} \label{arbitrary2}
	\sum\limits_{t=0}^{T-1} x(0)'{D'}^t\left(Q+L_0'RL_0\right)D^tx(0) + x(0)'{D'}^TP_0D^Tx(0).
	\end{equation}
	Note that because of stability $\eigmax{D}<1$, we have $$\lim\limits_{T \to \infty}x(0)'{D'}^TP_0D^Tx(0) =0.$$
	Therefore, by combining \eqref{arbitrary1}, \eqref{arbitrary2}, and \eqref{ricatti3},
	\begin{eqnarray*}
		&& x(0)'P_\infty \left(0\right) x(0) \\
		&=& \lim\limits_{T \to \infty}x(0)'P_T\left(0\right) x(0) \leq \lim\limits_{T \to \infty} x(0)'P_T\left(P_0\right) x(0) \\
		&\leq& \lim\limits_{T \to \infty} \sum\limits_{t=0}^{T-1} x(0)'{D'}^t\left(Q+L_0'RL_0\right)D^tx(0) \\
		&+& x(0)'{D'}^TP_0D^Tx(T) \\
		&=& \lim\limits_{T \to \infty} \sum\limits_{t=0}^{T-1} x(0)'{D'}^t\left(K_0-D'K_0D\right)D^tx(0) \\
		&=& x(0)'K_0x(0),
	\end{eqnarray*}
	i.e. for an arbitrary $P_0$, $$\lim\limits_{T \to \infty} P_T\left(P_0\right)=P_\infty(0).$$
	Using this, we show that $K_0$ is the unique solution of \eqref{ricatti2}. If $P_*$ is another solution, let $P_0=P_*$, which plugging in \eqref{Pdefine} implies that $P_t\left(P_*\right)=P_*$, for all $t=1,2, \cdots$. Hence $P_* = \lim\limits_{T \to \infty} P_T\left(P_*\right)=P_\infty(0)$, i.e. the solution $K_0$ of \eqref{ricatti2} exists, and is unique. 
	
	Conversely, if $K_0$ is a solution of \eqref{ricatti2}, define $L_0$ as \eqref{ricatti1} and $D=A_0+B_0L_0$. Note that $K_0$ is positive semidefinite, and let $P_0=K_0$. Defining $P_t$ by \eqref{Pdefine}, we obtain $P_t=K_0$, for all $t=0,1,\cdots$. As before, \eqref{ricatti2}, \eqref{ricatti1} imply \eqref{ricatti3}. Similarly, \eqref{stable1}, \eqref{stable2} hold, i.e. $\lim\limits_{t \to \infty} D^tx(0)=0$ for arbitrary $x(0)$, which implies that $L_0$ defined in \eqref{ricatti1} is a stabilizer.
\end{proof}
\section{\edit{Proof of Lemma \ref{SNeighborhood}}}
\begin{proof}
	Since $\para$ is stabilizable, according to Proposition \ref{stabilizable}, $\para \extendedLmatrix{\para}$ is stable; \begin{equation*}
	\eigmax{\para \extendedLmatrix{\para}} \leq 1-2\rho,
	\end{equation*}
	for some $\rho>0$. For arbitrary fixed $1 \leq i \leq p, 1 \leq j \leq q$, let all entries of the matrix $X_{ij} \in \R^{p \times q}$ be zero, except the $ij$-th entry, which is one. Then, for $\varphi \in \R$, consider the polynomial $$f_\varphi \left(\lambda\right)=\det \left( \left(\para+ \varphi X_{ij}\right)\extendedLmatrix{\para}-\lambda I_p \right).$$ 
	All coefficients of $f_\varphi \left(\lambda\right)$ are linear functions of $\varphi$. Further, the magnitudes of the roots of $f_\varphi \left(\lambda\right)$ are continuous with respect to the coefficients, and so, are also continuous with respect to $\varphi$. Since all roots of $f_0 \left(\lambda\right)$ are in magnitude at most $1-2 \rho$, there exists $\epsilon_{ij}>0$, such that $\left|\varphi\right|< \epsilon_{ij}$ implies that all roots of $f_\varphi \left(\lambda\right)$ are in magnitude at most $1-\left(2-1/(pq)\right) \rho$. Taking $\epsilon_0 = \min \limits_{i,j} \epsilon_{ij}$, by $\Mnorm{\para-\para_0}{2}<\epsilon_0$, $\para_0$ can be written in the form of $\para_0=\para+ \sum\limits_{i=1}^{p} \sum\limits_{j=1}^{q} \varphi_{ij}X_{ij}$, where $\left|\varphi_{ij}\right|<\epsilon_{ij}$, for all $i,j$. Therefore, all roots of $$f \left(\lambda\right)=\det \left( \para_0 \extendedLmatrix{\para}-\lambda I_p \right)$$ are in magnitude at most $1-\rho$, which is the desired result.
\end{proof}


\end{document}